\documentclass[aps,twocolumn,superscriptaddress]{revtex4-1}
\usepackage{graphicx}
\usepackage{amsmath}
\usepackage{amssymb}
\usepackage{amsthm}
\usepackage{color}
\usepackage{bm}

\newcommand{\ket}[1]{|{#1}\rangle}
\newcommand{\bra}[1]{\langle{#1}|}

\newcommand{\Tr}{\mathop{\text{Tr}}\nolimits}
\newcommand{\beq}{\begin{equation}}
\newcommand{\eeq}{\end{equation}}

\renewcommand{\Re}{\mathop{\text{Re}}}
\newcommand{\Span}{\mathop{\text{span}}}
\def\e{\mathrm{e}}
\def\i{\mathrm{i}}

\newtheorem{lemma}{Lemma}

\definecolor{dgreen}{rgb}{0,0.5,0}

\definecolor{delete}{cmyk}{0.5,0,0,0}

\begin{document}
\title{Universal Control Induced by Noise}
\author{Christian Arenz} 
\affiliation{Institute of Mathematics, Physics, and Computer Science, Aberystwyth University, Aberystwyth SY23 2BZ, UK}
\author{Daniel Burgarth} 
\affiliation{Institute of Mathematics, Physics, and Computer Science, Aberystwyth University, Aberystwyth SY23 2BZ, UK}
\author{Paolo Facchi}
\affiliation{Dipartimento di Fisica and MECENAS, Universit\`{a} di Bari, I-70126 Bari, Italy}
\affiliation{ INFN, Sezione di Bari, I-70126 Bari, Italy}
\author{Vittorio Giovannetti}
\affiliation{NEST, Scuola Normale Superiore and Istituto Nanoscienze-CNR, I-56126 Pisa, Italy}
\author{Hiromichi Nakazato}
\affiliation{Department of Physics, Waseda University, Tokyo 169-8555, Japan}
\author{Saverio Pascazio}
\affiliation{Dipartimento di Fisica and MECENAS, Universit\`{a} di Bari, I-70126 Bari, Italy}
\affiliation{ INFN, Sezione di Bari, I-70126 Bari, Italy}
\author{Kazuya Yuasa}
\affiliation{Department of Physics, Waseda University, Tokyo 169-8555, Japan}
\date{\today}

\begin{abstract}
On the basis of the quantum Zeno effect it has been recently shown [D. K. Burgarth \textit{et~al.}, Nat.\ Commun.\ \textbf{5}, 5173 (2014)] that a strong amplitude damping process applied locally on a part of a quantum system can have a beneficial effect on the dynamics of the remaining part of the system. 
Quantum operations that cannot be implemented without the dissipation become achievable by the action of the strong dissipative process. 
Here we generalize this idea by identifying decoherence-free subspaces (DFS's) as the subspaces in which the dynamics becomes more complex. 
Applying methods from quantum control theory we characterize the set of reachable operations within the DFS's. 
We provide three examples which become fully controllable within the DFS's while the control over the original Hilbert space in the absence of dissipation is trivial. 
In particular, we show that the (classical) Ising Hamiltonian is turned into a Heisenberg Hamiltonian by strong collective decoherence, which provides universal quantum computation within the DFS's.
Moreover we perform numerical gate optimization to study how the process fidelity scales with the noise strength.
As a byproduct a subsystem fidelity which can be applied in other optimization problems for open quantum systems is developed. 
\end{abstract}
\maketitle

\section{Introduction} 
The interaction of a quantum system with its environment is usually considered to be detrimental for quantum information processing. Quantum features one wants to use for quantum information tasks are washed out quickly so that the implementation of quantum gates becomes noisy. In the last decades, however, it has been observed that sometimes noise can be \textit{beneficial}. Rather than fighting against the environment, dissipative state preparation~\cite{CiracZoller, KrausZoller, MeGiovanna, Plenio} and dissipative quantum computing~\cite{WolfCirac, EisertWolf, BlattZoller} turned out to be valuable alternatives to unitary gate designs. In the context of quantum control theory state preparation and the implementation of unitary gates through the modulation of classical control fields in the presence of a dissipative environment have been studied~\cite{SchmidtCalarco, GlaserKhaneja, SchirmerFloether, KochReich} and the set of reachable operations has been analyzed~\cite{OpenQSC1, MeDaniel}. The environment can be used as a resource to increase the set of operations that can be implemented through the controls \cite{DanielZeno,KochReich2}. If the dissipative process admits some set of states robust against the environmental perturbations, the fidelity for the implementation of a gate within the subspaces spanned is not influenced by the noise and the dynamics there is free from decoherence. The existence of the decoherence-free subspaces (DFS's)~\cite{ZanardiRasetti1, ZanardiRasetti2, ref:Zanardi-PRA1998, LidarWhaley1, LidarWhaley2, ref:LidarBaconKempeWhaley-PRA2000, ref:BaconLidarDFS-PRL2000, ref:Zanardi-PRA2000, ref:BaconLidarDFS-PRA2001, Lidar} and the interplay between weak coherent processes and fast relaxation processes make it possible to implement unitary gates over the steady-state manifold in a noiseless manner~\cite{Zanardi1, Zanardi2, Beige, Oreshkov}. Here we show that such a noise process can even raise the fidelity for implementing a desired gate. The action of the strong dissipation allows the implementation of gate operations which cannot be realized without the help of the dissipation. The complexity of the dynamics is enhanced by the noise.

To show this we build upon the recent results obtained in Ref.~\cite{DanielZeno}. On the basis of the quantum Zeno effect~\cite{PascazioFacchi1} it was shown that frequent projective measurements can enrich the dynamics steered by a set of control Hamiltonians. Consider two control Hamiltonians $H_1$ and $H_2$ which are commutative with each other,
\begin{equation}
[H_1,H_2]=0.
\end{equation}
One is allowed to switch them on and off at will, but can induce only trivial dynamics on the system due to the commutativity.  If one additionally performs frequent projective measurements described by a Hermitian projection $P$ during the control, the system is confined to the subspace specified by the projection $P$ due to the quantum Zeno effect (quantum Zeno subspace~\cite{PascazioFacchi1, PascazioFacchi2}), where the system evolves unitarily (quantum Zeno dynamics~\cite{PascazioFacchi1, artzeno}) according to the projected counterparts of the control Hamiltonians, $PH_1P$ and $PH_2P$.
These projected Hamiltonians do not necessarily commute any more,
\begin{equation}
[PH_1P,PH_2P]\neq0.
\end{equation}
The measurement forces the system to evolve within the Zeno subspace, in which more complex operations can be realized thanks to the noncommutativity. The same effect can be induced by an infinitely strong dissipative process~\cite{Zanardi1,Zanardi2}.
It was shown in Ref.~\cite{DanielZeno} that a strong amplitude damping channel acting only locally on one out of many qubits in a chain typically turns a pair of commuting Hamiltonians into a pair of projected Hamiltonians that allow us to perform universal quantum computation over the whole chain of qubits apart from the projected one.
The amplitude damping acting locally on one qubit out of many, however, is a very special type of noise, and the assumption that it acts only locally seems unrealistic. 
On the other hand, this effect, \textit{noise-induced universal quantum computation}, should arise in more general settings.

In this article, we show that the universal controllability over the system can be achieved with the help of more general noise models, including the ones widely studied in the context of DFS's~\cite{ZanardiRasetti1, ZanardiRasetti2, ref:Zanardi-PRA1998, LidarWhaley1, LidarWhaley2, ref:LidarBaconKempeWhaley-PRA2000, ref:Zanardi-PRA2000, ref:BaconLidarDFS-PRA2001, Lidar, Zanardi1, Zanardi2}.
DFS's will be identified as the equivalent to the quantum Zeno subspaces.
Even if we are originally able to perform only trivial controls by commuting control Hamiltonians, a strong amplitude damping process projects the system onto DFS's, where we achieve universal controllability over the system. 
We characterize the set of reachable operations within DFS's and provide examples for which universal sets of gates can be implemented. Moreover, we perform numerical gate optimization to study how strong the dissipative process needs to be to implement such gates with high precision. As a byproduct a new fidelity function which can be applied in other optimization problems for open quantum systems is developed.

\section{Basic Concepts}
\label{sec:dynamicald}
\subsection{DFS's}
DFS's can be exploited as a passive strategy for protecting quantum information against noise~\cite{DCBook}. The theory has been developed in terms of interaction Hamiltonians~\cite{ZanardiRasetti1, ZanardiRasetti2, ref:Zanardi-PRA1998, ref:LidarBaconKempeWhaley-PRA2000, ref:BaconLidarDFS-PRL2000} as well as of quantum dynamical semigroups~\cite{LidarWhaley1, LidarWhaley2, Lidar, ref:BaconLidarDFS-PRA2001}. Many experiments, such as~\cite{White, Steinberg, Weinberg, Kwiat}, demonstrate the importance of DFS's for noiseless quantum computation. 
An experimental setup in waveguide QED has also been discussed recently \cite{ref:Kimble} and we will comment on it in Sec. \ref{sec:Numerical gate optimization}.
Moreover the combinations with error correcting schemes~\cite{LidarWhaley2} and dynamical strategies for decoherence control~\cite{PhysRevA.58.2733,PhysRevLett.82.2417,Zanardi199977,PhysRevA.66.012307,ref:BBZeno,ref:ControlDecoZeno} are promising possibilities for robust quantum information processing~\cite{WuLidar}.

A DFS can be seen as a degenerate pointer basis, which is invariant against the dissipative process. Consider a purely dissipative dynamics described by the Lindbladian generator
\begin{align}
\label{eq:Lindbladian}
\mathcal{D}(\rho)=-\sum_{j=1}^{d^{2}-1}\gamma_{j}(L_{j}^{\dagger}L_{j}\rho+\rho L_{j}^{\dagger}L_{j}-2L_{j}\rho L_{j}^{\dagger}),
\end{align}
with $\rho$ the density operator of the system, $L_{j}$ the Lindblad operators acting on the system, and $\gamma_{j}$ non-negative constants. Here we restrict ourselves to a finite-dimensional quantum system with Hilbert space $\mathcal H$ of dimension $d$ and write $S(\mathcal H)$ for the state space of $\mathcal H$. A DFS $\mathcal H_{\text{DFS}}^{(i)}\subset \mathcal H$ is spanned by $\{\ket{\psi_{1}^{(i)}},\ldots,\ket{\psi_{d_i}^{(i)}}\}$ characterized by
\begin{multline}
L_{j}\ket{\psi_{k}^{(i)}}=\lambda_{j}^{(i)}\ket{\psi_{k}^{(i)}},\quad
G\ket{\psi_{k}^{(i)}}=b^{(i)}\ket{\psi_{k}^{(i)}}
\\
(j=1,\ldots,d^{2}-1;\,k=1,\ldots,d_i),
\end{multline}
with $G=\sum_{j=1}^{d^{2}-1}\gamma_{j}L_{j}^{\dagger}L_{j}$, $\lambda_j^{(i)}$ complex, and $b^{(i)}=\sum_{j=1}^{d^{2}-1}\gamma_{j}|\lambda_j^{(i)}|^2$~\cite{Whaley}. Clearly if we prepare the system in an initial state $\rho_0\in S(\mathcal H_{\text{DFS}}^{(i)})$, this state is protected from dissipation driven by the dissipator $\mathcal{D}$ in~\eqref{eq:Lindbladian}. 
We denote by 
$\mathcal P$ the (super)projection (which is not necessarily self-dual) onto the steady-state manifold which consists of all quantum states $\rho$ satisfying $\mathcal D(\rho)=0$. 
We assume that the steady states are attractive, i.e.,
\begin{align}
\label{eq:superprocetor}
\lim_{t\to\infty}\e^{\mathcal{D}t}=\mathcal P, 
\end{align} 
to which we refer as the long-time/strong-damping limit. In practice, the strong dissipative process quickly destroys the quantum coherence along a given set of directions. 

\subsection{Quantum Control}
Having introduced the concept of DFS's we briefly review some results from quantum control theory. Consider a quantum system described by a Hamiltonian $H_0$, which suffers from dissipation described by the dissipator $\mathcal{D}$ in~\eqref{eq:Lindbladian}.
We try to steer the system by modulating external fields $\{f_{1}(t),\ldots,f_{m}(t)\}$ to switch on and off control Hamiltonians $\{H_{1},\ldots,H_{m}\}$. The evolution of the system is generated by 
\begin{equation}
\label{eq:totalgenerator}
\mathcal L_t(\rho)=-\i[H(t),\rho]+\mathcal D(\rho),
\end{equation}
with
\begin{equation}
H(t)=H_{0}+\sum\limits_{\ell=1}^{m}f_\ell(t)H_\ell.
\label{eqn:HT}
\end{equation}
$H_{0}$ is a drift Hamiltonian, and we do not have access to it. It is known~\cite{BookDalessandro} that in the absence of the dissipator $\mathcal D$, every unitary operation in the closure of the \textit{dynamical Lie group} $\e^{\mathfrak L}$ can be implemented with arbitrarily high precision, with 
\begin{equation}
\mathfrak L=\mathfrak{Lie}(\i H_{0},\i H_{1},\ldots,\i H_m)
\end{equation}
being the real Lie algebra formed by real linear combinations of the operators $\i H_0, \i H_1,\ldots,\i H_m$ and of their iterated commutators.  If $\mathfrak L\supseteq\mathfrak{su}(d)$ (for traceless operators), where $\mathfrak{su}(d)$ is the special unitary algebra, the system is said to be fully controllable, that is, every unitary can be implemented up to a global phase.

\section{Noise-Induced Universal Quantum Computation}
Our question is the following.
Suppose that the Lie algebra $\mathfrak{L}$ generated by our Hamiltonians $\{H_0,H_1,\ldots,H_m\}$ is strictly smaller than $\mathfrak{su}(d)$ and only limited unitaries are realizable by our control in the absence of the dissipation $\mathcal{D}$.
How is the set of reachable operations enlarged by the action of a strong dissipation $\mathcal D$ on the system?

To this end we need to know how the system evolves under the influence of the strong dissipation $\mathcal{D}$~\cite{Zanardi1, Zanardi2}.
To begin with we consider the situation in which no drift term $H_{0}$ is present and the dissipator $\mathcal D$ can be switched on and off arbitrarily as well as the control Hamiltonians $\{H_1,\ldots,H_m\}$. Afterwards we discuss the case in which we have no control over the dissipative part $\mathcal{D}$ and the drift Hamiltonian $H_0$, assuming that the control fields are all constant. Finally this leads to the general case~\eqref{eq:totalgenerator}.

If we are allowed to control $\mathcal D$ arbitrarily, we can switch rapidly between $\mathcal P$ and a unitary evolution that is generated by $\mathcal K_{c}=-\i [H_{c},{}\bullet{}]$ with some $H_{c}\in\{H_0,H_1,\ldots,H_m\}$ and in the limit of infinitely frequent switching
\begin{equation}
\label{eq:zeno}
\lim_{n\to\infty}(\mathcal P \e^{\mathcal K_c t/n}\mathcal P)^{n}=\e^{\mathcal P \mathcal K_c \mathcal P t }\mathcal P.
\end{equation}
It can be shown~\cite{Victor,Zanardi1} that
\begin{equation}
\label{eq:HermetianPDFS}
(\mathcal P\mathcal K_{c} \mathcal P)(\rho)=-\i [P_{i}H_{c}P_{i},\rho],\quad\forall \rho\in S(\mathcal H_{\text{DFS}}^{(i)}),
\end{equation}
where $P_{i}=\sum_{k=1}^{d_i}\ket{\psi_k^{(i)}}\bra{\psi_k^{(i)}}$ is the Hermitian projection on the $i$th DFS\@. Clearly this implies that if we prepare the system in a DFS, say in the $i$th DFS, it remains there evolving unitarily with the projected Hamiltonian $P_{i}H_{c}P_{i}$. Furthermore if the evolution generated by $\mathcal D$ is unital, i.e., $\mathcal D(\openone)=0$, the system evolves over the steady state manifold according to $\mathcal P\mathcal K_{c}\mathcal P=-\i [\mathcal P(H_{c}),{}\bullet{}]$, and for an Abelian interaction algebra~\cite{Kribs}, generated by the $L_{j}$'s in~\eqref{eq:Lindbladian} and their conjugates, we have $\mathcal P(H_{c})=\sum_{i}P_{i}H_{c} P_{i}$~\cite{Zanardi1}. 
The mechanism is similar to that of the quantum Zeno subspaces induced by other means, such as frequent measurements, strong continuous couplings, and frequent unitary kicks~\cite{PascazioFacchi1,PascazioFacchi2,ref:ControlDecoZeno}. The projective measurement is effectively performed by the dissipative process.
The measurement is nonselective~\cite{schwinger}: the transitions among different subspaces are hindered and the dynamics within each subspace is governed by the projected Hamiltonian $P_{i}H_{c}P_{i}$.

So far we have discussed the case in which the dissipator $\mathcal{D}$ as well as the control Hamiltonians $\{H_1,\ldots,H_m\}$ can be controlled arbitrarily, in the absence of the drift Hamiltonian $H_0$. Typically one has no access to the dissipative part $\mathcal{D}$ in~\eqref{eq:totalgenerator} that arises for example from an interaction with the environment. If we assume that the control fields are all constant, the generator~\eqref{eq:totalgenerator} including the drift Hamiltonian $H_0$ reads 
\begin{equation}
\mathcal L=g\mathcal K+\mathcal D,
\end{equation} 
where we have introduced the constant $g$ that measures the strength of the coherent part $\mathcal K=-\i [H,{}\bullet{}]$ in comparison with the dissipative part $\mathcal{D}$. Based on a perturbative expansion it has been shown~\cite{Zanardi1, Zanardi2} that 
\begin{equation}
\label{eq:ProjectionTHM}
\|(\e^{t\mathcal L}-\e^{gt\mathcal P\mathcal K \mathcal P})\mathcal P\|\leq O(g\tau_{R}),
\end{equation}
where $\tau_{R}^{-1}=\min_{h>0}|{\Re}\{\lambda_{h}\}|$, with $\lambda_{h}$ the nonvanishing eigenvalues of $\mathcal D$, defines the longest relaxation time scale $\tau_{R}$. The norm is the usual operator norm and $gt=O(1)$. Thanks to this, we notice that on a time scale on which the dissipative dynamics is much faster than the coherent dynamics, the dynamics is effectively governed by~\eqref{eq:zeno}. Similarly to~\eqref{eq:HermetianPDFS}, if the system is initially prepared in a DFS, say in the $i$th DFS, the system evolves unitarily within the same DFS in the limit $g\tau_{R}\to 0$ with $gt=O(1)$, driven by the projected Hamiltonian $P_{i}HP_{i}$. 
Again, this is intuitively clear: if the dynamics is dominated by the fast dissipative process, the latter defines the subspaces within which the system can evolve. The presence of the coherent component $\mathcal K$ only modifies the motion within each subspace.

It is now easy to treat the general case~\eqref{eq:totalgenerator}. In the spirit of the Trotter formula, by switching among the control Hamiltonians under $g\tau_{R}\to 0$ and $gt=O(1)$, we can implement with arbitrarily high precision every $U_i=\e^{\mathfrak{L}_{\text{DFS}}^{(i)}}$ in the relevant DFS, with 
\begin{equation}
\label{eq:LiealgDFS}
\mathfrak{L}_{\text{DFS}}^{(i)}=\mathfrak{Lie}(\i P_{i}H_{0}P_{i},\i P_{i}H_{1}P_{i},\ldots,\i P_{i}H_mP_{i})
\end{equation}
being the real Lie algebra generated by the drift Hamiltonian $H_0$ and the control Hamiltonians $\{H_1,\ldots,H_m\}$ projected by the projection $P_i$. Note that for a unital evolution $\e^{\mathcal D t}$ the Lie algebra over the DFS's reads 
\begin{align}
\mathfrak L_{\text{DFS}}=\mathfrak{Lie}\bm{(}\i\mathcal P(H_{0}),\i\mathcal P(H_{1}),\ldots,\i\mathcal P(H_{m})\bm{)}. 	
\end{align}

The projection $P_{i}$ can now be identified as the equivalent of the frequent projective measurement that projects the system onto the quantum Zeno subspace specified by $P_i$: the strong dissipation does the same job as the Zeno measurement. In the strong-damping limit the system is confined in the DFS's, evolving unitarily and steered by the projected Hamiltonians.

Although the dimensions of the DFS's are smaller than the dimension of the original Hilbert space, the dynamics induced by the projected control Hamiltonians within the DFS's can be much more complex than the one induced by the original control Hamiltonians in the absence of the dissipation, since $\dim\mathfrak L_{\text{DFS}}$ is in general larger than $\dim\mathfrak L$~\cite{DanielZeno}. One can even achieve the universal controllability over the DFS's, with the help of the strong dissipation.

\section{Universal control in DFS's: Examples}
On the basis of the observation that the projected drift and control Hamiltonians do not necessarily commute any more, we saw in the last section that the Lie algebra over the DFS's might be larger than the Lie algebra over the original Hilbert space.  In the following we present three different examples, for which the universal controllability over the DFS's is achieved, even though only ``simple'' operations can be implemented over the original Hilbert space in the absence of dissipation.

\subsection{Two Qubits}
\label{ExamplesTwoQ}
We first provide a simplest example with only two qubits, which is essentially the same as that presented in Ref.~\cite{DanielZeno}: one of the two qubits, say qubit 2, is subject to a strong amplitude-damping process.
We also discuss the same model but with a pure dephasing process on qubit 2, instead of the amplitude-damping process.

The drift Hamiltonian reads 
 \begin{subequations}
 \begin{equation}
 \label{eq:driftHamiltonianTwoQ}
 H_{0}=\sigma_{x}\otimes (\sigma_{x}+\sigma_{z}),
 \end{equation}
 while we have a control Hamiltonian 
 \begin{equation}
 H_{1}=\sigma_{y}\otimes (\sigma_{x}-\sigma_{z}),
 \end{equation}
 \end{subequations}
where $\sigma_\alpha$ ($\alpha=x,y,z$) are the Pauli operators.
Note that these Hamiltonians commute with each other, $[H_{0}, H_{1}]=0$.
Therefore in the absence of noise the Lie algebra $\mathfrak{L}=\mathfrak{Lie}(\i H_0,\i H_1)$ is spanned just by $\{\i H_{0},\i H_{1}\}$ and hence is only two dimensional, $\dim\mathfrak{L}=2$. We now add amplitude-damping on qubit 2, generated by 
 \begin{equation}
 \label{eq:amplitudedamping2}
 \mathcal D(\rho)=-\gamma(\sigma_{+}^{(2)}\sigma_{-}^{(2)}\rho+\rho\sigma_{+}^{(2)}\sigma_{-}^{(2)}-2\sigma_{-}^{(2)}\rho\sigma_{+}^{(2)}),
 \end{equation}
with $\sigma_{\pm}^{(2)}=\openone\otimes(\sigma_x\pm \i\sigma_y)/2$ the raising and lowering operators acting nontrivially only on qubit 2. 
It projects the system as~\cite{ref:SolutionLindblad}
\begin{align}
\e^{\mathcal{D}t}\rho
&=(P+Q\e^{-\gamma t})\rho(P+Q\e^{-\gamma t})
+(1-\e^{-2\gamma t})L\rho L^\dag
\nonumber\\
&\xrightarrow{\gamma t\to\infty}
 \mathcal P(\rho)=P\rho P+L\rho L^{\dagger},
\label{eqn:PTwoQubitAmp}
\end{align}
where $P=\openone\otimes\ket{0}\bra{0}$, $Q=\openone\otimes\ket{1}\bra{1}$, and $L=\sigma_-^{(2)}=\openone\otimes\ket{0}\bra{1}$ with $\ket{0}$ and $\ket{1}$ being the eigenstates of  $\sigma_{z}$ belonging to the eigenvalues $-1$ and $+1$, respectively.
The dissipator~(\ref{eq:amplitudedamping2}) admits a single DFS identified by the Hermitian projection $P$ onto
\begin{equation}
\mathcal{H}_\text{DFS}
=\Span\{\ket{0}\otimes\ket{0},\ket{1}\otimes\ket{0}\}.
\end{equation} 
In the strong-damping limit our Hamiltonians are projected to
\begin{subequations}
\begin{gather}
PH_0P=-\sigma_x\otimes\ket{0}\bra{0},\\
PH_1P=\sigma_y\otimes\ket{0}\bra{0},
\end{gather}
\end{subequations}
and the Lie algebra over the DFS is given by 
 \begin{equation}
 \mathfrak L_{\text{DFS}}
=\mathfrak{Lie}(\i PH_0P,\i PH_1P)
=\mathfrak{su}(2)\otimes\ket{0}\bra{0}.
\end{equation}
That is, in the strong-damping limit qubit 1 becomes fully controllable, i.e., every $U\in \text{SU}(2)$ can be implemented on qubit 1.

Now let us replace the amplitude-damping process on qubit 2 by a pure dephasing process generated by
\begin{equation}
\label{eq:dephasing}
\mathcal D(\rho)=-\gamma[\sigma_{z}^{(2)},[\sigma_{z}^{(2)},\rho]],
\end{equation}
where $\sigma_{z}^{(2)}=\openone\otimes\sigma_{z}$. In this case the system is projected as~\cite{ref:SolutionLindblad}
\begin{align}
\e^{\mathcal{D}t}\rho
&=
P_0\rho P_0
+P_1\rho P_1
+P_0\rho P_1\e^{-4\gamma t}
+P_1\rho P_0\e^{-4\gamma t}
\nonumber\\
&\xrightarrow{\gamma t\to\infty}
\mathcal P(\rho)=P_{0}\rho P_{0}+P_{1}\rho P_{1},
\end{align}
where $P_{i}=\openone\otimes\ket{i}\bra{i}$ ($i=0,1$).
This dephasing process admits two orthogonal DFS's identified by the Hermitian projections $P_0$ and $P_1$,
\begin{subequations}
\label{eq:DFSdephasing}
\begin{align}
\mathcal H_{\text{DSF}}^{(0)}&=\Span\{\ket{0}\otimes\ket{0},\ket{1}\otimes \ket{0}\},\\
\mathcal H_{\text{DFS}}^{(1)}&=\Span\{\ket{0}\otimes \ket{1},\ket{1}\otimes \ket{1}\}.
\end{align}
\end{subequations}
Since the evolution generated by~\eqref{eq:dephasing} is unital, in the strong-dephasing limit our Hamiltonians are projected to
\begin{subequations}
\begin{gather}
\mathcal{P}(H_0)=\sigma_x\otimes\sigma_z,\\
\mathcal{P}(H_1)=-\sigma_y\otimes\sigma_z,
\end{gather}
\end{subequations}
and the Lie algebra over the DFS's $\mathfrak L_{\text{DFS}}
=\mathfrak{Lie}\bm{(}\i\mathcal{P}(H_0),\i\mathcal{P}(H_1)\bm{)}$ is spanned by $\{\sigma_x\otimes\sigma_z,\sigma_y\otimes\sigma_z,\sigma_z\otimes\openone\}$: its dimension is $\dim\mathfrak{L}_\text{DFS}=3$ and is increased from $\dim\mathfrak{L}=2$ by the action of the strong pure dephasing on qubit 2.
In particular, if qubit 2 starts from the state $\ket{i}$ ($i=0$ or $1$) the Lie algebra over the $i$th DFS reads 
\begin{equation}
 \mathfrak L_{\text{DFS}}^{(i)}
=\mathfrak{Lie}(\i P_i H_0P_i,\i P_i H_1P_i)
=\mathfrak{su}(2)\otimes\ket{i}\bra{i},
\end{equation}
and qubit 1 is fully controllable.
Although in this case we do not have the full controllability over all DFS's, universal quantum computation is possible on qubit 1 within either of the two DFS's. We see that using the framework of DFS's the previous results on amplitude damping channels extend naturally to other types of noise.

\subsection{$\bm{N}$-Level Atom with an Unstable Level}
The next example involves an atom with energy eigenstates $\ket{1},\ldots,\ket{N}$ plus a higher lying unstable state $\ket{e}$ that decays to the lower lying states with rates $\gamma_{1},\ldots,\gamma_{N}$, as schematically represented in Fig.\ \ref{fig: schematicalRepNAT}.
We assume that $N\ge2$.
A similar level structure manifests for example in a Rydberg atom, for which the quantum Zeno dynamics has recently been demonstrated in an impressive way~\cite{Haroche}. 
  \begin{figure}[t]
  \includegraphics[width=0.9\columnwidth]{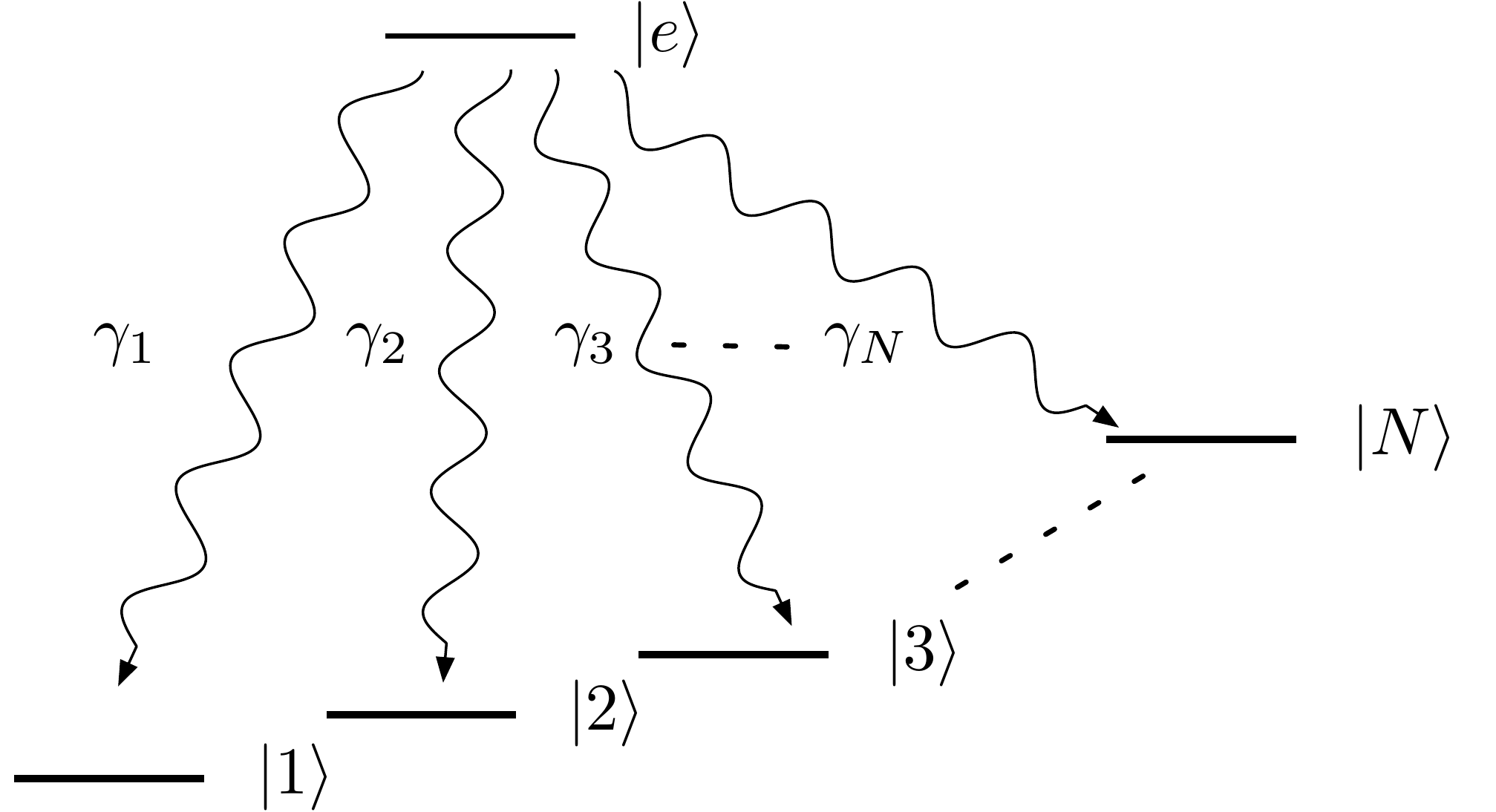}
  \caption{\label{fig: schematicalRepNAT} Schematic representation of an $N$-level atom with a higher lying unstable level $\ket{e}$ that decays with rates $\gamma_{1},\ldots,\gamma_{N}$ to the lower lying levels $\ket{1},\ldots,\ket{N}$ spanning a DFS\@.}
\end{figure}

We will consider a decay process described by
\begin{equation}
\label{eq:lindbladianNL}
\mathcal D(\rho)=-\sum_{j=1}^{N}\gamma_{j}(L_j^{\dagger}L_j\rho+\rho L_j^{\dagger}L_j-2L_j\rho L_j^\dag)
\end{equation}
with $L_j=\ket{j}\bra{e}$ ($j=1,\ldots,N$). 
The system is projected as~\cite{ref:SolutionLindblad}
\begin{align}
\e^{\mathcal{D}t}\rho
&=
(P+Q\e^{-\Gamma t})
\rho
(P+Q\e^{-\Gamma t})
\nonumber\\
&\qquad\qquad\,
{}+\frac{1}{\Gamma}(1-\e^{-2\Gamma t})
\sum\limits_{j=1}^{N}\gamma_{j}L_{j}\rho L_{j}^{\dagger}
\nonumber\\
&\xrightarrow{\Gamma t\to\infty}
\mathcal P(\rho)=P\rho P+\frac{1}{\Gamma}\sum\limits_{j=1}^{N}\gamma_{j}L_{j}\rho L_{j}^{\dagger},
\label{eq:superprojectionNLevel}
\end{align}
where $P=\openone-\ket{e}\bra{e}$, $Q=\ket{e}\bra{e}$, and $\Gamma=\sum_{j=1}^N\gamma_j$.
The dissipator~(\ref{eq:lindbladianNL}) admits a DFS identified by the Hermitian projection $P$, namely, spanned by the lower lying levels 
\begin{equation}
\mathcal{H}_\text{DFS}
=\Span\{\ket{1},\ldots,\ket{N}\}.
\end{equation}

Now we are going to introduce a drift Hamiltonian and a control Hamiltonian. We take an example from Ref.~\cite{DanielPas}, for which the universal control is achieved through frequent projective measurements described by a Hermitian projection $P$. Note that here $P$ is realized through the strong-damping limit of the completely positive and trace-preserving (CPTP) map that is generated by the dissipator~\eqref{eq:lindbladianNL}.  The drift Hamiltonian 
\begin{subequations}
\begin{equation}
H_{0}=\ket{e}\bra{2}+\ket{2}\bra{e}
+\sum_{j=1}^{N-1}(\ket{j}\bra{j+1}+\ket{j+1}\bra{j})
\end{equation}
consists of the interactions among the lower lying levels $\{\ket{1},\ldots,\ket{N}\}$ and additional driving terms stimulating the transitions between $\ket{e}$ and $\ket{2}$. The control Hamiltonian, on the other hand, reads
\begin{equation}
H_{1}=\ket{e}\bra{e}+\ket{1}\bra{1}-(\ket{e}\bra{1}+\ket{1}\bra{e}).
\end{equation}
\end{subequations}
Again, these Hamiltonians commute with each other, $[H_{0}, H_{1}]=0$.
Therefore in the absence of the noise $\mathcal{D}$ the Lie algebra $\mathfrak{L}=\mathfrak{Lie}(\i H_0,\i H_1)$ is spanned just by $\{\i H_{0},\i H_{1}\}$ and hence is only two dimensional, $\dim\mathfrak{L}=2$, as in the previous example.
These Hamiltonians are projected by the strong dissipation~(\ref{eq:superprojectionNLevel}) to
\begin{subequations}
\begin{gather}
PH_{0}P=\sum_{j=1}^{N-1}(\ket{j}\bra{j+1}+\ket{j+1}\bra{j}),\\
PH_{1}P=\ket{1}\bra{1}.
\end{gather}
\end{subequations}
This pair of Hamiltonians is known to generate the full unitary algebra $\mathfrak u(N)$ (see e.g.~\cite{ref:QSI}).
We get
\begin{equation}
 \mathfrak L_{\text{DFS}}
=\mathfrak{Lie}(\i PH_0P,\i PH_1P)
=\mathfrak u(N) P.
\end{equation} 
Its dimension is $\dim\mathfrak L_{\text{DFS}}=N^{2}$, while $\dim\mathfrak L=2$ in the absence of the dissipation. Compared to the previous two-qubit example we observe here a more dramatic increase of the complexity in the dynamics over the DFS through projection.

\subsection{Ising Chain of $\bm{N}$ Qubits under Collective Decoherence}
\label{sec:Chain}
The third example is a chain of $N$ qubits 
interacting with each other via nearest-neighbor Ising-type couplings,
\begin{subequations}
\label{eqn:Ising}
\begin{equation}
H_0
=\sum_{n=1}^{N-1}\sigma_z^{(n)}\sigma_z^{(n+1)},
\label{eqn:Ising0}
\end{equation}
where $\sigma_\alpha^{(n)}=\openone\otimes\cdots\otimes\openone\otimes\sigma_\alpha\otimes\openone\otimes\cdots\otimes\openone$ ($\alpha=x,y,z$) are the Pauli operators acting on the $n$th qubit.
We assume that $N\ge3$.
In addition we are allowed to switch on and off the coupling between the first two qubits,
\begin{equation}
H_1=\sigma_z^{(1)}\sigma_z^{(2)}.
\label{eqn:Ising1}
\end{equation}
\end{subequations}
These Hamiltonians trivially commute with each other, $[H_0,H_1]=0$, and our control over the chain of qubits is very poor.
Suppose then that this system undergoes a strong collective decoherence described by the Lindbladian generator
\begin{equation}
\mathcal{D}(\rho)=-\sum_{\alpha=x,y,z}\gamma_\alpha(S_\alpha^2\rho+\rho S_\alpha^2-2S_\alpha\rho S_\alpha),
\label{eqn:CollectiveDeco}
\end{equation}
that is unital, where
\begin{equation}
S_\alpha=\frac{1}{2}\sum_{n=1}^N\sigma_\alpha^{(n)}
\quad
(\alpha=x,y,z)
\end{equation}
are the collective spin operators. 
This noise model is well studied in the context of DFS's, and is known to admit multiple DFS's labeled by the total spin $J$ of the whole chain [i.e., $J$ gives the total spin angular momentum of the chain by $\bm{S}^2=\sum_{\alpha=x,y,z}S_\alpha^2=J(J+1)$]~\cite{ref:BaconLidarDFS-PRL2000,ref:BaconLidarDFS-PRA2001,Zanardi1}.
The dimensions of the DFS's are given by~\cite{ref:MandelWolf} 
\begin{equation}
d_{J,N}=\frac{(2J+1)N!}{(N/2+J+1)!(N/2-J)!},
\label{eqn:dJN}
\end{equation}
and are listed in Table \ref{tab:dJN} for small numbers of qubits $N$.

To see how our Hamiltonians $H_0$ and $H_1$ are projected by the collective decoherence $\Lambda_t=\e^{\mathcal{D}t}$ in the strong-damping limit, let us look at its dual channel $\Lambda_t^\star=\e^{\mathcal{D}^\star t}$ defined by
\begin{equation}
\Tr\{A\Lambda_t(\rho)\}
=\Tr\{\Lambda_t^\star(A)\rho\},
\end{equation}
for an arbitrary observable $A$ and state $\rho$, and note that $\mathcal{D}^\star=\mathcal{D}$ in this case, since $S_\alpha$ in the generator $\mathcal{D}$ in~(\ref{eqn:CollectiveDeco}) are Hermitian.
By this channel, each component of our Hamiltonians $\sigma_z^{(n)}\sigma_z^{(n+1)}$ evolves according to
\begin{widetext}
\begin{equation}
\mathcal{D}
\begin{pmatrix}
\sigma_x^{(n)}\sigma_x^{(n+1)}\\
\sigma_y^{(n)}\sigma_y^{(n+1)}\\
\sigma_z^{(n)}\sigma_z^{(n+1)}
\end{pmatrix}
=-2\begin{pmatrix}
\gamma_y+\gamma_z&-\gamma_z&-\gamma_y\\
-\gamma_z&\gamma_z+\gamma_x&-\gamma_x\\
-\gamma_y&-\gamma_x&\gamma_x+\gamma_y
\end{pmatrix}
\begin{pmatrix}
\sigma_x^{(n)}\sigma_x^{(n+1)}\\
\sigma_y^{(n)}\sigma_y^{(n+1)}\\
\sigma_z^{(n)}\sigma_z^{(n+1)}
\end{pmatrix},
\label{eqn:DstarMatrix}
\end{equation}
\end{widetext}
and in the strong-damping limit the operators $\sigma_\alpha^{(n)}\sigma_\alpha^{(n+1)}$ ($\alpha=x,y,z$) are projected to

\begin{align}
\Lambda_t
\begin{pmatrix}
\sigma_x^{(n)}\sigma_x^{(n+1)}\\
\sigma_y^{(n)}\sigma_y^{(n+1)}\\
\sigma_z^{(n)}\sigma_z^{(n+1)}
\end{pmatrix}
\xrightarrow{\bar{\gamma}t\to\infty}{}&
\mathcal{P}
\begin{pmatrix}
\sigma_x^{(n)}\sigma_x^{(n+1)}\\
\sigma_y^{(n)}\sigma_y^{(n+1)}\\
\sigma_z^{(n)}\sigma_z^{(n+1)}
\end{pmatrix}
\nonumber\\
={}&
\frac{1}{3}
\begin{pmatrix}
\bm{\sigma}^{(n)}\cdot\bm{\sigma}^{(n+1)}\\
\bm{\sigma}^{(n)}\cdot\bm{\sigma}^{(n+1)}\\
\bm{\sigma}^{(n)}\cdot\bm{\sigma}^{(n+1)}
\end{pmatrix},
\end{align}
where $\bar{\gamma}$ is a characteristic timescale of the decoherence, e.g., the smaller nonvanishing eigenvalue of the matrix in~(\ref{eqn:DstarMatrix}).
The operators become rotationally symmetric by the projection.
In particular, our Hamiltonians $H_0$ and $H_1$ are projected to
\begin{subequations}
\label{eqn:Heisenberg}
\begin{gather}
\mathcal{P}(H_0)
=\frac{1}{3}\sum_{n=1}^{N-1}\bm{\sigma}^{(n)}\cdot\bm{\sigma}^{(n+1)},\\
\mathcal{P}(H_1)
=\frac{1}{3}\bm{\sigma}^{(1)}\cdot\bm{\sigma}^{(2)}.
\end{gather}
\end{subequations}
The Ising chain~(\ref{eqn:Ising}) thus becomes the Heisenberg chain~(\ref{eqn:Heisenberg}) by the projection $\mathcal{P}$.
The projected Hamiltonians are not commutative anymore with each other.

Now we look at the Lie algebra
\begin{equation}
\mathfrak L_{\text{DFS}}
=\mathfrak{Lie}\bm{(}\i\mathcal{P}(H_0),\i\mathcal{P}(H_1)\bm{)}
\end{equation}
generated by the projected Hamiltonians $\mathcal{P}(H_0)$ and $\mathcal{P}(H_1)$.
Recall that the projected Hamiltonians in~(\ref{eqn:Heisenberg}) are rotationally symmetric, reflecting the character of the decoherence model~(\ref{eqn:CollectiveDeco}).
Commutators preserve this rotational symmetry, as we will see below.
Then, all the elements of the Lie algebra $\mathfrak L_{\text{DFS}}$ are rotationally symmetric, and are given in terms of the two- and three-body operators (see Appendix \ref{app:ChainLemmaLiealg} for details)
\begin{align}
&H_{mn}
=\bm{\sigma}^{(m)}\cdot\bm{\sigma}^{(n)},\quad
H_{ijk}
=\bm{\sigma}^{(i)}\cdot(\bm{\sigma}^{(j)}\times\bm{\sigma}^{(k)})
\nonumber\\
&\qquad\,%
(m<n;\ i<j<k;\ m,n,i,j,k=1,\ldots,N).
\label{eqn:TwoThreeBody}
\end{align}

In Ref.~\cite{ref:BaconLidarDFS-PRA2001}, it is proved that any $\text{SU}$ transformations on the DFS's induced by the strong collective decoherence~(\ref{eqn:CollectiveDeco}) can be realized if we are able to apply \textsc{swap} interactions between any pair of qubits.
Note that the \textsc{swap} Hamiltonians can be constructed from the rotationally symmetric two-body operators $H_{mn}=\bm{\sigma}^{(m)}\cdot\bm{\sigma}^{(n)}$: the \textsc{swap} operator $S_{mn}$ swapping the states of qubits $m$ and $n$ is given by $S_{mn}=(1+\bm{\sigma}^{(m)}\cdot\bm{\sigma}^{(n)})/2$.
Since we have proven in Appendix \ref{app:ChainLemmaLiealg} that all the rotationally symmetric two-body operators $H_{mn}=\bm{\sigma}^{(m)}\cdot\bm{\sigma}^{(n)}$ can be generated by the projected Hamiltonians $\mathcal{P}(H_0)$ and $\mathcal{P}(H_1)$, the \textsc{swap} Hamiltonians $S_{mn}$ between any pair of qubits can be applied, and by the theorem proved in Ref.~\cite{ref:BaconLidarDFS-PRA2001} all the generators  of $\bigoplus_J\mathfrak{su}(d_{J,N})$ can be constructed.
Namely, 
\begin{equation}
\mathfrak L_{\text{DFS}}
=\mathfrak{Lie}\bm{(}\i\mathcal{P}(H_0),\i\mathcal{P}(H_1)\bm{)}
\supset\bigoplus_J\mathfrak{su}(d_{J,N}).
\label{eqn:LieChainSU}
\end{equation}
This means that we are able to perform universal quantum computation over all DFS's by the projected Hamiltonians $\mathcal{P}(H_0)$ and $\mathcal{P}(H_1)$.

\begin{table}[b]
\caption{The dimensions $d_{J,N}$ of the DFS's, and the dimension of the Lie algebra $\mathfrak{L}_\text{DFS}=\mathfrak{Lie}\bm{(}\i\mathcal{P}(H_0),\i\mathcal{P}(H_1)\bm{)}$ compared with the dimensions of the $\mathfrak{u}$ and $\mathfrak{su}$ algebras over the DFS's, for small numbers of qubits $N$.}
\label{tab:dJN}
\begin{tabular}{r@{\ }cccccc}
\hline
\hline
&$N=1$&$N=2$&$N=3$&$N=4$&$N=5$&$N=6$\\
$J=0$&&1&&2&&5\\
$J=\frac{1}{2}$&1&&2&&5\\
$J=1$&&1&&3&&9\\
$J=\frac{3}{2}$&&&1&&4\\
$J=2$&&&&1&&5\\
$J=\frac{5}{2}$&&&&&1\\
$J=3$&&&&&&1\\
$\dim\mathfrak{L}_\text{DFS}$
&0&1&4&12&40&129\\
$\sum_J\dim\mathfrak{su}(d_{J,N})$
&0&0&3&11&39&128\\
$\sum_J\dim\mathfrak{u}(d_{J,N})$
&1&2&5&14&42&132\\
\hline
\hline
\end{tabular}
\end{table}

Notice, however, that the full unitary algebra $\bigoplus_J\mathfrak{u}(d_{J,N})$ over the DFS's is not attainable.
For instance, not all the rotationally symmetric four-body operators 
$
(\bm{\sigma}^{(i)}\cdot\bm{\sigma}^{(j)})
(\bm{\sigma}^{(k)}\cdot\bm{\sigma}^{(\ell)})
=H_{ij}H_{k\ell}
$
can be generated.
Combinations of them can be generated by the rotationally symmetric two- and three-body operators through
\begin{equation}
\i [H_{ij},H_{jk\ell}]=2(H_{ik}H_{j\ell}-H_{\i\ell}H_{jk}),
\end{equation}
but we realize that we can generate only differences of four-body operators.
The other commutators such as
\begin{equation}
\i [H_{ijk},H_{ij}H_{k\ell}]
=4(H_{j\ell}-H_{\i\ell})
+2(H_{\i\ell}H_{jk}-H_{ik}H_{j\ell}),
\end{equation}
do not help to break the differences to get a single piece of four-body operator.
This is because commutators yield something antisymmetric with respect to some of the qubits involved in the operators. 
In order to single out each piece of four-body operator from the differences, we need a 
sum of four-body operators, but it is not available or provided through commutators.
We thus cannot generate the full algebra over the DFS's.

See Table \ref{tab:dJN}, where the dimension of the Lie algebra 
$\dim\mathfrak{L}_\text{DFS}$ is compared with the dimension of the $\mathfrak{su}$ algebra $\sum_J\dim\mathfrak{su}(d_{J,N})$ and that of the full unitary algebra $\sum_J\dim\mathfrak{u}(d_{J,N})$ over the DFS's.
The dimension of the Lie algebra $\dim\mathfrak{L}_\text{DFS}$ is indeed larger than $\sum_J\dim\mathfrak{su}(d_{J,N})$, but is smaller than $\sum_J\dim\mathfrak{u}(d_{J,N})$.
Anyway, the dimension of the Lie algebra is greatly enhanced from $\dim\mathfrak{L}=2$, as $\dim\mathfrak{L}_\text{DFS}\simeq4^NN^{-3/2}/\sqrt{\pi}$ for large $N$, as estimated in Appendix \ref{app:LieChain}\@.

In summary, we started with  two commuting Hamiltonians $H_0$ and $H_1$ in (\ref{eqn:Ising}), which are projected to $\mathcal{P}(H_0)$ and $\mathcal{P}(H_1)$ in~(\ref{eqn:Heisenberg}), respectively, by the strong collective decoherence~(\ref{eqn:CollectiveDeco}).
As a consequence, the Ising chain~(\ref{eqn:Ising}) is changed into the Heisenberg chain~(\ref{eqn:Heisenberg}), and our projected Hamiltonians $\mathcal{P}(H_0)$ and $\mathcal{P}(H_1)$ are not commutative anymore with each other.
They generate the full algebra of $\bigoplus_J\mathfrak{su}(d_{J,N})$ on the DFS's. Remarkably the noise is turning the Ising chain (classical) into the Heisenberg chain (quantum), and we are able to perform a universal quantum computation over the DFS's.

\section{Gate optimization and subsystem fidelity}
\label{sec:Numerical gate optimization}
In this section we analyze how the process fidelity scales with the noise strength. To this end we resort to the numerical gate optimization using the quantum control package implemented in QuTip~\cite{webpageQuTip}. We study the two-qubit example discussed in Sec.\ \ref{ExamplesTwoQ}, with the amplitude damping~\eqref{eq:amplitudedamping2} for different values of $\gamma$. For the sake of simplicity the drift Hamiltonian~\eqref{eq:driftHamiltonianTwoQ} is treated as a control Hamiltonian as well.

We wish to optimize the control fields $f_\ell(t)$ [recall~(\ref{eqn:HT})] to implement some goal operation $\mathcal E_{G}$. Denote by $\mathcal E_{T}=\mathop{\text{T}}\exp[\int_{0}^{T}dt^{\prime}\,\mathcal L(t^{\prime})]$ the CPTP map at time $T$, where $\mathcal{L}(t)$ is the Liouvillian given in~(\ref{eq:totalgenerator}) and $\mathop{\text{T}}$ indicates time-ordered product. The optimization is performed to minimize the gate error
\begin{equation}
 \label{eq:gateerror}
 \varepsilon_{1}=\|\mathcal E_{T}-\mathcal E_{G} \|_{\text{HS}}^{2},
\end{equation}
where $\|{}\bullet{}\|_{\text{HS}}$ is the Hilbert-Schmidt norm with $\mathcal E_{G}$ and $\mathcal E_{T}$ being treated as $d^{2}\times d^{2}$ matrices obtained by the row-vectorization of the density operator of a $d$-dimensional system. In general, for two CPTP maps $\Phi_{1}$ and $\Phi_{2}$, the Hilbert-Schmidt norm of the difference between their corresponding matrices 
provides an upper bound $\|\Phi_{1}-\Phi_{2}\|_{\diamond}\leq d\|\Phi_{1}-\Phi_{2}\|_{\text{HS}}$ on the diamond norm $\|{}\bullet{}\|_{\diamond}$. The diamond norm~\cite{Diamond} takes its maximal value $2$ when the two quantum channels $\Phi_{1}$ and $\Phi_{2}$ are perfectly distinguishable. The minimization of~\eqref{eq:gateerror} is done by a gradient-based algorithm~\cite{Grape} dividing the total time $T$ into equidistant time intervals, on which the control fields are piecewise constant.

We are actually interested in the reduced dynamics of system 1, i.e., in the map $\mathcal E_{T}^{(1)}(\rho_{1})=\Tr_{2}\{\mathcal E_{T}(\rho_{1}\otimes\rho_{2})\}$ with $\rho_{1}$ and $\rho_{2}$ the initial states of systems 1 and 2, respectively, and $\Tr_{2}$ the partial trace over system 2. We wish to optimize $\mathcal E_{T}$ such that $\mathcal E_{T}^{(1)}$ becomes some goal unitary map $\mathcal E_{G}^{(1)}=\mathcal U_{G}$ with $\mathcal U_{G}(\rho)=U_{G}\rho U_{G}^\dag$ and $U_{G}\in\text{SU}(d)$. 
Our measure of error $\varepsilon_1$ in~(\ref{eq:gateerror}), however, depends also on how the channels $\mathcal{E}_T$ and $\mathcal{E}_G$ act on system 2: even if $\mathcal{E}_T^{(1)}$ coincides with the goal unitary $\mathcal E_{G}^{(1)}=\mathcal U_{G}$, the total maps $\mathcal{E}_T$ and $\mathcal{E}_G$ can be different and our measure of error $\varepsilon_1$ can be nonvanishing.
In addition, the reduced map $\mathcal{E}_T^{(1)}$ depends on the initial state of system 2. 
We notice, on the other hand, that since the goal operation on system 1 is unitary $\mathcal{U}_G$ the total goal operation must factorize $\mathcal E_{G}=\mathcal U_{G}\otimes \tilde{\mathcal E}$ with $\tilde{\mathcal E}$ an arbitrary CPTP map acting on system 2.
What is more relevant is how close the reduced channel $\mathcal{E}_T^{(1)}$ is to the goal unitary $\mathcal{U}_G$.
Therefore it would be more appropriate to perform an additional minimization of $\varepsilon_1$ in~(\ref{eq:gateerror}) over $\tilde{\mathcal E}$.
To obtain the subsystem fidelity for purely unitary channels this minimization can be carried out analytically~\cite{UnitaryM1, UnitaryM2} but unfortunately for arbitrary CPTP channels this is a challenging task. 
Instead we use the normalized Choi representation $J(\mathcal E)$ of a quantum channel $\mathcal E$~\cite{WolfGT} to derive a lower bound of $\varepsilon_1$,
\begin{align}
\label{eq:epsilon2}
\varepsilon_1/d^2&
=\|J(\mathcal E_{T})-S(J(\mathcal U_{G})\otimes J(\tilde{\mathcal E}))S   \|_{\text{HS}}^{2} \nonumber \\
&\geq \Tr\{J^{2}(\mathcal E_{T})(\openone- S(J(\mathcal U_{G})\otimes\openone_{2})S)\}  
\equiv \varepsilon_{2},
\end{align}
where the \textsc{swap} operator $S$ between systems 1 and 2 is introduced because in general for two CPTP maps $\Phi_{1}$ and $\Phi_{2}$, $J( \Phi_{1}\otimes \Phi_{2})=S (J(\Phi_{1})\otimes J(\Phi_{2}))S$. 
For details of the derivation of the lower bound~\eqref{eq:epsilon2} we refer to Appendix \ref{lowerbound}. Clearly the minimization over $\tilde{\mathcal E}$ on the left-hand side of~\eqref{eq:epsilon2} is now lower bounded by $\varepsilon_{2}$, which is independent of $\tilde{\mathcal E}$ and is zero if and only if the goal unitary operation on system 1 is reached. Thus the lower bound becomes tighter and tighter when $\mathcal E_{T}$ factorizes into the goal unitary $\mathcal{U}_G$ on system 1 and some arbitrary $\tilde{\mathcal E}$ on system 2.

\begin{figure}[t]
 \includegraphics[width=0.90\columnwidth]{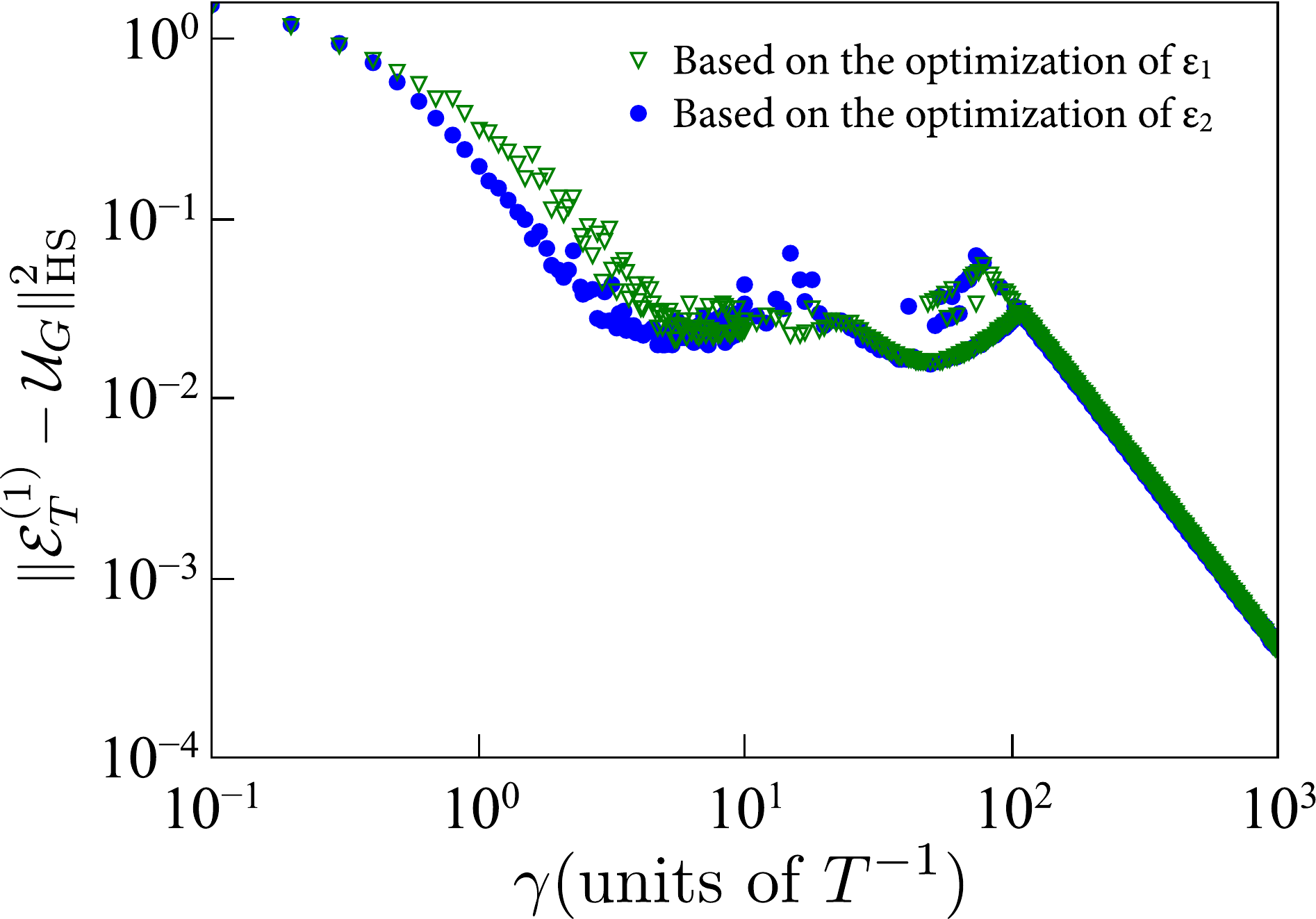}
 \caption{\label{fig:fid} (Colour online) Numerical gate optimization for the two-qubit model in Sec.\ \ref{ExamplesTwoQ} with the amplitude damping~\eqref{eq:amplitudedamping2} for different values of $\gamma$.
 The gate error between the reduced dynamics $\mathcal E_{T}^{(1)}$ and the Hadamard gate on qubit 1 obtained from the numerical minimizations of $\varepsilon_{1}$ (green triangles) and $\varepsilon_{2}$ (blue points) for different values of $\gamma$ with gate time $T=1$. Qubit 2 is initially prepared in the totally mixed state, and for $\varepsilon_{1}$, $\tilde{\mathcal E}$ is chosen to be the superprojection $\mathcal{P}$ that brings qubit 2 into the ground state $\ket{0}$. To reduce the effect of local minima in the minimum value 100 randomly chosen initial pulses are taken.}
\end{figure}

The strategy to study the convergence of the map to the goal operation as $\gamma$ is increased can now be summarized as follows. We implement $\varepsilon_{2}$ and its gradient with respect to the control fields on QuTip, and minimize $\varepsilon_{1}$ and $\varepsilon_{2}$ for different values of $\gamma$. For $\varepsilon_{1}$, $\tilde{\mathcal E}$ is chosen to be the superprojection $\mathcal{P}$ in~(\ref{eqn:PTwoQubitAmp}) that brings qubit 2 into the ground state $\ket{0}$. 
On the basis of the minimizations of $\varepsilon_{1}$ and $\varepsilon_{2}$ we evaluated in Fig.\ \ref{fig:fid} the gate error $\|\mathcal E_{T}^{(1)}-\mathcal U_{G}\|_{\text{HS}}^{2}$ by specifying the initial state of qubit 2 in the totally mixed state and tracing out the auxiliary degrees of freedom. The target unitary operation $U_{G}$ on qubit 1 was chosen to be the Hadamard gate. We observe that despite the enhanced freedom in $\varepsilon_{2}$ the curves based on the minimizations of $\varepsilon_1$ and $\varepsilon_2$ are similar to each other. For noise strengths above $\gamma \approx 10\,T^{-1}$ gate errors below $10^{-1}$ can be reached, corresponding to the upper bound $0.2$ for the diamond norm. 
It demonstrates that with intermediate noise strengths reasonable fidelity can be reached.

Besides being fundamentally interesting we now want to discuss in more detail the experimental feasibility of this observation. Together with controlling commuting interactions, the main ingredient of the observed behavior is a strong dissipative process and the emergence of DFS's. Thus, in practice, we need a system containing a subset of states that are stable on an appropriate time scale and a dissipative process decaying into this subset, while being much faster than other noise processes. An attractive platform that provides such a noise process is waveguide QED, i.e., the interaction of quantum emitters with the modes of a waveguide, such as photonic crystal waveguides \cite{Goban}, optical fibers \cite{Mitsch}, and superconducting circuits \cite{Mlynek}. For further details regarding waveguide QED we refer to \cite{ref:Kimble} and references therein. In particular, in such systems the presence of collective decoherence described by a Lindbladian of the form \eqref{eqn:CollectiveDeco} gives rise to DFS's and moreover the high density of modes of the waveguide yields regions in which large decay rates are achieved. Recently, the ability to implement universal gates in such systems over a DFS was studied in detail in \cite{ref:Kimble}. While in this study weak driving fields with a constant envelope were used to implement a target unitary gate, we remark that there is no fundamental restriction of using time-dependent controls to implement the ideas that are proposed here similarly. Indeed, it was shown that as long as the distance between the quantum emitters in the waveguide is small, the gate error $\varepsilon$ for implementing a specific gate over the DFS scales as $\varepsilon\propto1/\sqrt{F}$ \cite{ref:Kimble}. Here $F=\gamma/\gamma^*$ is the Purcell factor given by the ratio of the decay rate $\gamma$ into the DFS and the decay rate $\gamma^{*}$ of other noise channels,  such that a unitary gate can accurately be implemented if the Purcell factor is reasonably large. Together with the ability to individually address transitions of the embedded quantum emitters waveguide QED systems provide therefore a promising platform for noise-induced universal quantum computation.

\section{Conclusions}
We showed that every dissipative process exhibiting a DFS can enlarge the set of unitary operations that can be implemented by means of classical control fields. We provided three examples for which a universal set of gates can be implemented over a DFS whereas over the original Hilbert space only ``simple'' operations are possible. In particular we showed that a realistic noise model can map a commutative classical system into a universal quantum one. Numerical gate optimization was performed to study how strong the dissipative process needs to be to implement some unitary gate over the DFS with high precision. As a result a subsystem fidelity for open quantum systems was developed. 
Our results pave the way to experimental feasibility studies in noisy systems such as quantum emitters in a waveguide.

\acknowledgments
We thank Thomas-Schulte Herbr\"uggen for fruitful discussions, HPC Wales for providing a high performance computer cluster on which the numerical simulations were done, and Alexander James Pitchford for advice regarding the numerical simulations and developing QuTip. CA acknowledges financial support from a HPC Wales bursary and the Aberystwyth Doctoral Career Development Scholarship. DB acknowledges support from the EPSRC grant EP/M01634X/1. 
KY is supported by the Grant-in-Aid for Scientific Research (C) (No.\ 26400406) from the Japan Society for the Promotion of Science (JSPS) and by the Waseda University Grants for Special Research Projects (No.\ 2015K-202 and No.\ 2016K-215).
This work was also supported by the Top Global University Project from the Ministry of Education, Culture, Sports, Science and Technology (MEXT), Japan, by the Italian National Group of Mathematical Physics (GNFM-INdAM), by INFN through the project ``QUANTUM,'' and by PRIN Grant No.\ 2010LLKJBX ``Collective quantum phenomena: from strongly correlated systems to quantum simulators.''

\appendix
\section{Characterization of $\bm{\mathfrak L_{\text{DFS}}}$ for the Qubit Chain Model}
\label{app:ChainLemmaLiealg}
We prove the following lemma:
\begin{lemma}
The Lie algebra $\mathfrak L_{\text{DFS}}$ generated by the two projected Hamiltonians $\mathcal{P}(H_0)$ and $\mathcal{P}(H_1)$ in~(\ref{eqn:Heisenberg}) includes all the rotationally symmetric two- and three-body operators, $H_{mn}=\bm{\sigma}^{(m)}\cdot\bm{\sigma}^{(n)}$ and $H_{ijk}=\bm{\sigma}^{(i)}\cdot(\bm{\sigma}^{(j)}\times\bm{\sigma}^{(k)})$ ($m<n$; $i<j<k$; $m,n,i,j,k=1,\ldots,N$) defined in~(\ref{eqn:TwoThreeBody}), for any number of qubits $N\ge3$.
\end{lemma}
\begin{proof}
Let us introduce
\begin{equation}
\tilde{H}_0
=\mathcal{P}(H_0),\qquad
\tilde{H}_1
=\mathcal{P}(H_1).
\end{equation}
The first commutator reads
\begin{equation}
\i [\tilde{H}_0,\tilde{H}_1]
=2H_{123}.
\end{equation}
Then, by commuting $\tilde{H}_1=H_{12}$ with the newly generated $H_{123}$ twice, we have
\begin{subequations}
\begin{align}
\i [H_{12},H_{123}]
&=4(H_{13}-H_{23}),
\\
\i\bm{[}\i[H_{12},H_{123}],H_{123}\bm{]}
&=16(H_{13}+H_{23}-2H_{12}),
\end{align}
\end{subequations}
from which we gain $H_{13}$ and $H_{23}$.
All the rotationally symmetric operators up to the third qubit (three two-body operators $H_{12}$, $H_{23}$, $H_{13}$ and a three-body operator $H_{123}$) are in our hands.

For $N\ge4$, we proceed by induction.
Suppose that all the rotationally symmetric two- and three-body operators for the first $n$ qubits are at our disposal.
It is actually the case for $n=3$, as we saw above.
Then, we are able to extend one qubit further, generating all the two- and three-body operators involving the $(n+1)$th qubit by the following procedure.

\bigskip
\noindent
1. Commute $H_{(n-1)n}$ with $\tilde{H}_0$ to extend to the $(n+1)$th qubit,
\begin{equation}
\i [H_{(n-1)n},\tilde{H}_0]
=-2(
H_{(n-2)(n-1)n}
-
H_{(n-1)n(n+1)}
).
\end{equation}
We acquire $H_{(n-1)n(n+1)}$.

\bigskip
\noindent
2. By commuting $H_{(n-1)n}$ with the newly generated $H_{(n-1)n(n+1)}$ twice, we have
\begin{subequations}
\begin{align}
&\i [H_{(n-1)n},H_{(n-1)n(n+1)}]
=4(H_{(n-1)(n+1)}-H_{n(n+1)}),\\
&\i\bm{[}\i [H_{(n-1)n},H_{(n-1)n(n+1)}],H_{(n-1)n(n+1)}\bm{]}
\nonumber\\
&\qquad
=16(H_{(n-1)(n+1)}+H_{n(n+1)}-2H_{(n-1)n}),
\end{align}
\end{subequations}
from which we gain $H_{(n-1)(n+1)}$ and $H_{n(n+1)}$.

\bigskip
\noindent
3. Then, iterate the following steps for $m=n-2,n-3,\ldots,1$,
\begin{subequations}
\begin{align}
\i [
H_{m(m+1)},H_{(m+1)(n+1)}
]
&=2H_{m(m+1)(n+1)},
\\
\i [
H_{m(m+1)},H_{m(m+1)(n+1)}
]
&=4(H_{m(n+1)}-H_{(m+1)(n+1)}),
\end{align}
\end{subequations}
to get $H_{m(n+1)}$ ($m=1,\ldots,n-2$).
All the two-body operators involving the $(n+1)$th qubit are thus in our hands.

\bigskip
\noindent
4. Combining the two-body operators, we can generate any three-body operators involving the $(n+1)$th qubit,
\begin{multline}
\i [
H_{m_1m_2},H_{m_2(n+1)}
]
=2H_{m_1m_2(n+1)}\\
(m_1,m_2=1,\ldots,n;\ m_1<m_2\le n).
\end{multline}

\bigskip
\noindent
In this way, all the rotationally symmetric two- and three-body operators for the first $n+1$ qubits are generated.
Then, by induction, we can generate all the rotationally symmetric two- and three-body operators for any number of qubits $N$.
\end{proof}

\section{Asymptotic Dimension of the Lie Algebra $\bm{\mathfrak{L}_\text{DFS}}$ for the Qubit Chain Model}
\label{app:LieChain}
Let us estimate the asymptotic dimension for a large $N$ of the Lie algebra $\mathfrak{L}_\text{DFS}$ in~(\ref{eqn:LieChainSU}) generated by the projected Hamiltonians $\mathcal{P}(H_0)$ and $\mathcal{P}(H_1)$ for the chain of $N$ qubits discussed in Sec.\ \ref{sec:Chain}.
As commented in Sec.\ \ref{sec:Chain}, the dimension of $\mathfrak{L}_\text{DFS}$ is bounded by the dimension of $\bigoplus_J\mathfrak{su}(d_{J,N})$ and the dimension of $\bigoplus_J\mathfrak{u}(d_{J,N})$, i.e.,
\begin{equation}
\sum_J(d_{J,N}^2-1)
<\dim\mathfrak{L}_\text{DFS}
<\sum_Jd_{J,N}^2.
\end{equation}
As we will see, the lower bound is dominated by the first contribution $\sum_Jd_{J,N}^2$ for large $N$, and the difference between the lower and upper bounds becomes relatively negligible in the asymptotic regime.
Observe also that the dimensions $d_{J,N}$ of the DFS's given in~(\ref{eqn:dJN}) can be cast as
\begin{multline}
d_{J,N}
=\left(1-\frac{2K}{N+1}\right)
\begin{pmatrix}
\smallskip
N+1\\
K
\end{pmatrix}
\\
(K=N/2-J=0,1,\ldots,\lfloor\tfrac{N}{2}\rfloor),
\end{multline}
where $\lfloor x\rfloor$ denotes the largest integer not greater than $x$.
Approximating the binomial coefficient by
\begin{equation}
\begin{pmatrix}
\smallskip
n\\
k
\end{pmatrix}
=\frac{2^n}{\sqrt{\pi n/2}}\e^{-2n(k/n-1/2)^2}[
1+O(1/\sqrt{n})
],
\end{equation}
the dimension of the Lie algebra is estimated as
\begin{align}
\dim\mathfrak{L}_\text{DFS}
&\sim\sum_J d_{J,N}^2
\nonumber\\
&=\sum_{K=0}^{\lfloor N/2\rfloor}
\left(1-\frac{2K}{N+1}\right)^2
\begin{pmatrix}
\smallskip
N+1\\
K
\end{pmatrix}^2
\nonumber\\
&\sim
\frac{N+1}{2}
\int_0^1dx
\,x^2
\frac{4^{N+1}}{\pi(N+1)/2}\e^{-(N+1)x^2}
\nonumber\displaybreak[0]\\
&\sim
\frac{4^N}{\sqrt{\pi}N^{3/2}},
\end{align}
where the continuum limit is taken through $x=1-2K/(N+1)$.

\section{Derivation of the Lower Bound $\bm{\varepsilon_{2}}$}
\label{lowerbound}
Here we derive the lower bound~\eqref{eq:epsilon2}. Using the definition of the Hilbert-Schmidt norm $\|A\|_{\text{HS}}^{2}=\Tr\{A^{\dagger}A\}$ for a matrix $A$ we can rewrite the left-hand side of~\eqref{eq:epsilon2},
\begin{align}
&\|J(\mathcal E_{T})-S(J(\mathcal U_{G})\otimes J(\tilde{\mathcal E}))S   \|_{\text{HS}}^{2}\nonumber \\ 
&\quad
=\Tr\{J^{2}(\mathcal E_{T})\}+\Tr_2\{J^{2}(\tilde{\mathcal E})\}
\nonumber\\
&\quad\qquad
{}-2\Tr\{SJ(\mathcal E_{T})S(J(\mathcal U_{G})\otimes J(\tilde{\mathcal E}))\},
\label{eq:frobnorm}
\end{align}
where $\Tr_{2}$ denotes the partial trace over the second system and the properties of the normalized Choi state $J$ were used, i.e., $J^{\dagger}=J$,  $\Tr\{J\}=1$, and $J^{2}=J$ for a unitary map. The third term of the right-hand side of~\eqref{eq:frobnorm} can be rewritten as 
\begin{align}
&\Tr\{SJ(\mathcal E_{T})S(J(\mathcal U_{G})\otimes J(\tilde{\mathcal E}))\}\nonumber\\
&\quad
=\Tr\{SJ(\mathcal E_{T})S(J(\mathcal U_{G})\otimes\openone_{2})   (J(\mathcal U_{G})\otimes J(\tilde{\mathcal E}))\}\nonumber \\
&\quad
\leq \Tr\{SJ^{2}(\mathcal E_{T})S(J(\mathcal U_{G}\otimes\openone_{2}))\}^{1/2}\Tr_{2}\{J^{2}(\tilde{\mathcal E})\}^{1/2}\nonumber \\
&\quad
\leq \frac{1}{2}\left(\Tr\{J^{2}(\mathcal E_{T})S(J(\mathcal U_{G})\otimes\openone_{2})S\}    +\Tr_{2}\{J^{2}(\tilde{\mathcal E})\}\right),
\label{eq:inequality1}
\end{align}
where from the second line to the third the Cauchy-Schwarz inequality and from the third line to the fourth the inequality between the arithmetic and the geometric means have been used. Combining~\eqref{eq:frobnorm} and~\eqref{eq:inequality1} we arrive at 
\begin{align}
&\|J(\mathcal E_{T})-S(J(\mathcal U_{G})\otimes J(\tilde{\mathcal E}))S   \|_{\text{HS}}^{2}\nonumber \\ 
&\quad
\geq  \Tr\{J^{2}(\mathcal E_{T})\}+\Tr_{2}\{J^{2}(\tilde{\mathcal E})\}\nonumber \\
&\quad\qquad
{} - \Tr\{J^{2}(\mathcal E_{T})S(J(\mathcal U_{G})\otimes\openone_{2})S\}    -\Tr_{2}\{J^{2}(\tilde{\mathcal E})\}\nonumber \\
&\quad
=\Tr\{J^{2}(\mathcal E_{T})(\openone- S(J(\mathcal U_{G})\otimes\openone_{2})S)\},
\end{align}
which is the desired result. Note that for pure unitary maps $\mathcal E_{T}=\mathcal U_{T}$ the lower bound simplifies further 
\begin{align}
\varepsilon_{2}=1-\Tr\{J(\mathcal U_{T})S(J(\mathcal U_{G})\otimes\openone_{2})S\}.
\end{align}


\end{document}